\documentclass[runningheads]{llncs}

\usepackage{graphicx}

\usepackage[T1]{fontenc}
\usepackage[utf8]{inputenc}
\usepackage{lmodern}

\usepackage[english]{babel}
\usepackage{microtype}

\usepackage{hyperref}
\usepackage{xcolor}

\usepackage{amsmath,amsfonts,amssymb}
\usepackage{dsfont}

\usepackage{cite}
\usepackage[capitalise]{cleveref}

\usepackage{comment}

\newcommand{\N}{\mathbb{N}}

\newcommand{\Q}{\mathbb{Q}}
\newcommand{\R}{\mathbb{R}}

\newcommand{\BAfL}{\mathsf{BAfL}}

\newcommand{\AfL}{\mathsf{AfL}}
\newcommand{\AfLn}{\mathsf{AfL}^{\neq}}
\newcommand{\AfLe}{\mathsf{AfL}^{=}}
\newcommand{\AfLQ}{\mathsf{AfL}_{\mathbb{Q}}}
\newcommand{\AfLA}{\mathsf{AfL}_{\mathbb{A}}}

\newcommand{\SL}{\mathsf{SL}}
\newcommand{\SLQ}{\mathsf{SL}_{\mathbb{Q}}}
\newcommand{\SLn}{\mathsf{SL}^{\neq}}
\newcommand{\SLe}{\mathsf{SL}^=}
\newcommand{\SLeQ}{\mathsf{SL}^=_{\mathbb{Q}}}

\newcommand{\Log}{\mathsf{L}}
\newcommand{\PSPACE}{\mathsf{PSPACE}}
\newcommand{\sgn}{\operatorname{sgn}}

\newcommand{\re}{\operatorname{Re}}
\newcommand{\lcm}{\operatorname{lcm}}
\newcommand{\ord}{\operatorname{ord}}

\def\abs#1{\left\vert #1\right\vert}

\def\L{\ensuremath{\Log}}
\def\res{\operatorname{Res}}
\def\diag{\operatorname{diag}}
\def\E{\mathbb E}

\usepackage{enumitem}
\setlist[itemize]{topsep=0pt,itemsep=-1ex,partopsep=1ex,parsep=1ex}

\newcommand{\EquationSpaceGlobal}{6pt}
\AtBeginDocument{
  \setlength{\belowdisplayskip}{\EquationSpaceGlobal}
  \setlength{\belowdisplayshortskip}{\EquationSpaceGlobal}
  \setlength{\abovedisplayskip}{\EquationSpaceGlobal}
  \setlength{\abovedisplayshortskip}{\EquationSpaceGlobal}
}

\newtheorem{prop}[theorem]{Property}

\begin{document}

\title{Computational Limitations of Affine Automata}

\author{Mika Hirvensalo\inst{1} \and
Etienne Moutot\inst{1,2}\orcidID{0000-0003-2073-4709} \and
Abuzer Yakary\i lmaz\inst{3}\orcidID{0000-0002-2372-252X}}
\authorrunning{M. Hirvensalo, E. Moutot, and A, Yakary\i lmaz}
%
\institute{
Department of Mathematics and Statistics, University of Turku, FI-20014 Turku, Finland\\
\email{mikhirve@utu.fi}
\and
{LIP, ENS de Lyon – CNRS – UCBL – Universit\'e de Lyon
, \'Ecole Normale Sup\'erieure de Lyon, Lyon, France\\ \email{etienne.moutot@ens-lyon.org}}
\and
Center for Quantum Computer Science, Faculty of Computing\\University of Latvia, R\={\i}ga, Latvia \\
\email{abuzer@lu.lv}
}
\maketitle              

\begin{abstract}
    We present two new results on the computational limitations of affine automata. First, we show that the computation of bounded-error rational-values affine automata is simulated in logarithmic space. Second, we give an impossibility result for algebraic-valued affine automata. As a result, we identify some unary languages (in logarithmic space) that are not recognized by algebraic-valued affine automata with cutpoints.
\end{abstract}

\section{Introduction}

Finite automata are an interesting model to study since they express the very natural limitation of finite memory.
They are also good computational models, since they are simpler than many others machines like pushdown automata or Turing machines. Due to this simplicity, there exists many different models of finite automata, all trying to express different computational settings. Deterministic \cite{Sipser}, probabilistic \cite{Paz} and quantum \cite{Ambainis2015} finite automata (DFAs, PFAs, and QFAs, respectively) have been studied to try to understand better the computational limitations inherent to all these cases.

Recently, Díaz-Caro and Yakary\i lmaz introduced a new model, called {\em affine computation} \cite{Yaka16}. As a non-physical model, the goal of affine computation is to investigate the power of interference caused by negative amplitudes in the computation, like in the quantum case.
But unlike QFAs, affine finite automata (AfAs) have unbounded state set and the final operation corresponding to quantum measurement cannot be interpreted as linear.
The final operation in AfAs is analogous to renormalization in Kondacs-Watrous \cite{KW97} or Latvian \cite{Ambainis06} quantum automata models.

AfAs and their certain generalizations have been investigated in a series of works \cite{Yaka16,Yaka16-2,Yaka17,HiMoYa}. In most of the cases, affine models (e.g., bounded-error and unbouded-error AfAs, zero-error affine OBDDs, zero-error affine counter automata, etc.) have been shown more powerful than their classical or quantum counterparts. On the other hand, we still do not know too much regarding the computational limitations of AfAs. Towards this direction, we present two new results. First, we show that the computation of bounded-error rational-values affine automata is simulated in logarithmic space, and so we answer positively one of the open problems in \cite{Yaka16}. Second, we give an impossibility result for algebraic-valued AfAs, and, as a result, we identify some unary languages (in logarithmic space) that are not recognized by algebraic-valued AfAs with cutpoints.

\section{Preliminaries}
\label{section:def}

For a given word $ w $, $ w_i $ represents its $i$-th letter. For any given class $ \sf C $, $ \mathsf{C}_{\mathbb{Q}} $ and $ \mathsf{C}_{\mathbb{A}} $ denotes the classes defined by the machines restricted to have rational-valued and algebraic-valued components, respectively. The logarithmic and polynomial space classes are denoted as $ \Log $ and $\PSPACE$, respectively. We assume that the reader is familiar with the basics of automata theory.

\subsection{Models}
As a {\em probability distribution} (also known as a {\em stochastic vector}) we understand a (column) vector with nonnegative entries summing up to one, and a {\em stochastic matrix} (also known as a {\em Markov matrix}) here stands for a square matrix whose all columns are probability distributions.

\begin{definition}[PFA]
A $k$-state probabilistic finite automaton (PFA) $P$ over alphabet $\Sigma$ is a triplet
\[
P = (\vec{x},\{M_i\mid i\in\Sigma\},\vec{y})
\]
where $\vec{x}\in\mathbb R^k$ is a stochastic vector called {\em initial distribution}, each $M_i\in\mathbb R^{k\times k}$ is a stochastic matrix, and $\vec{y}\in\mathbb\{0,1\}^k$ is the final vector (each 1 in $\vec{y}$ represents an accepting state).
\end{definition}

For any input word $w \in \Sigma^*$ with length $n$, $P$ has a probability distribution of states as follows:
$
M_w \vec{x} = M_{w_n} \cdots M_{w_1} \vec{x}.
$
The {\em accepting probability} corresponds to the probability of $P$ being in an accepting state after reading $w$, which is given by
\begin{equation}
    \label{finalvalue01}
    f_P(w) = \vec{y}^TM_{w^R} \vec{x}.
\end{equation}


Affine finite automaton (AfA) is a generalization of PFA allowing negative transition values. Only allowing negative values in the transition matrices does not add any power (generalized PFAs are equivalent to PFAs, see \cite{Turakainen1969}), but affine automata introduce also a non-linear behaviour. The automaton acts like a generalized probabilistic automaton until the last operation, which is a non-linear operation called a \textit{weighting operation}. 

\begin{definition}
A vector $\vec{v}\in\R^k$ is an affine vector if and only if its coordinates sums up to $1$. A matrix $M$ is an affine matrix if and only if all its columns are affine vectors.
\end{definition}
The following property is straightforward to verify, and it will ensure that affine automata are well defined.

\begin{prop}
If $M$ and $N$ are affine matrices, then $MN$ is also an affine matrix. In particular, if $\vec{v}$ is an affine vector, then $M\vec{v}$ is also an affine vector.
\end{prop}

\begin{definition}[AfA]
  \label{def:AfA}
  A $k$-state {\em AfA} $A$ over alphabet $\Sigma$ is a triplet
  \[
	A = (\vec{x} ,\{M_i ~|~ i\in\Sigma\}, F)
	\]
where $\vec{x}$ is an initial affine vector, each $M_i$ is an affine transition matrix, and $F=\diag(\delta_1,\ldots,\delta_n)$ is the final projection matrix, where each $\delta_i\in\{0,1\}$ for $1\leq i \leq n$.
\end{definition}

The value computed by an affine automaton can be most conveniently be defined via the following notion:

\begin{definition}
Notation $|\vec{v}| = \sum_i |v_i|$ stands for the usual $L^1$ norm.
\end{definition}

Now, the final value of the affine automaton $A$ of Definition \ref{def:AfA} is
\begin{equation}
    \label{finalvalue03}
    f_A(w) = \frac{|F M_w\vec{v}_0|}{|M_w \vec{v}_0|}.
\end{equation}
Clearly $f_A(w)\in[0,1]$ for any input word $w\in\Sigma^*$.

\begin{remark}
    Notice that the final value for PFAs (\ref{finalvalue01}) is defined as matrix product $\vec{v}_f\mapsto\vec{y}^T\vec{v}_f$, which is a linear operation on $\vec{v}_f$. On the other hand, computing final value from $\vec{v}_f$ as in (\ref{finalvalue03}) involves nonlinear operations $\displaystyle\vec{v}_f\mapsto\frac{|F\vec{v}_f|}{|\vec{v}_f|}$ such as $L^1$-norm and normalization (division).
\end{remark}

\subsection{Cutpoint languages}
Given a function $f:\Sigma^*\to[0,1]$ computed by an automaton (stochastic or affine), there are different ways of defining the language of recognized by this automaton.

\begin{definition}[Cutpoint languages]
  \label{def:SL}
  A language $L\subseteq\Sigma^*$ is recognized by an automaton $A$ with cutpoint $\lambda \in [0,1) $ if and only if
  \[ L = \{ w\in\Sigma^* ~|~ f_A(w) > \lambda \} .\]
These languages are called cutpoint languages.
In the case of probabilistic (resp., affine automata), the set of cut-point languages are called {\em stochastic languages} (resp., {\em affine languages}) and denoted by $\SL$ (resp., $ \AfL $).
\end{definition}
We remark that fixing the cutpoint in the interval $ (0,1) $ does not change the classes $ \SL $ and $ \AfL $ \cite{Paz,Yaka16}.

\begin{definition}[Exclusive cutpoint languages]
  \label{def:SLn}
  A language $L\subseteq\Sigma^*$ is recognized by an automaton $A$ with exclusive cutpoint $\lambda \in [0,1] $ if and only if
  \[ L = \{ w\in\Sigma^* ~|~ f_A(w) \neq \lambda \} .\]
These languages are called exclusive cutpoint languages.
In the case of probabilistic (resp., affine automata), the set of exclusive cut-point languages are called {\em exclusive stochastic languages} (resp., {\em exclusive affine languages}) and denoted by $\SLn$ (resp., $ \AfLn $). The complement of $\SLn$ (resp., $ \AfLn $) is $ \SLe $ (resp., $\AfLe$).
\end{definition}
Again, we remark that fixing the cutpoint in the interval $ (0,1) $ does not change the classes $ \SLn $, $\SLe$, $ \AfLn $, and $\AfLe$ \cite{Paz,Yaka2010,Yaka16}.

A stronger condition is to impose that accepted and rejected words are separated by a gap: the cutpoint is said to be isolated.

\begin{definition}[Isolated cutpoint or bounded error]
\label{def:BSL}
A language $L$ is recognized by an automaton $A$ with {\em isolated cutpoint} $\lambda$ if and only if there exist $\delta>0$ such that $\forall w\in L, f_A(w) \geq \lambda+\delta$, and
$\forall w\notin L, f_A(w) \leq \lambda-\delta$.
The set of languages recognized with {\em bounded error} (or isolated cutpoint) affine automata is denoted by $\BAfL$.
\end{definition}

A classical result by Rabin \cite{Rabin} shows that isolated cutpoint stochastic languages are regular. Rabin's proof essentially relies on two facts: 1) the function mapping the final vector into $[0,1]$ is a contraction, and 2) the state vector set is bounded.
By modifying Rabin's proof, it is possible to show that also many quantum variants of stochastic automata obey the same principle \cite{Ambainis2015}: bounded-error property implies the regularity of the accepted languages. In fact, E. Jeandel generalized Rabin's proof by demonstrating that the compactness of the state vector set together with the continuity of the final function are sufficient to guarantee the regularity of the accepted language if the cutpoint is isolated \cite{Jeandel2007}.


\section{Logarithmic simulation}
\label{sec:log-simulation}

Macarie \cite{Macarie1998} proved that $\SLeQ \subseteq \Log $ and $\SLQ \subseteq\Log $. That is, the computation of any rational-valued probabilistic automaton can be simulated by an algorithm using only logarithmic space.
However, this logarithmic simulation cannot be directly generalized for rational-valued affine automata due to the non-linearity of their last operation.
In order to understand why, we will first reproduce the proof.

Before that, let us introduce the most important space-saving technique:

\begin{definition} Notation $(b\mod c)$ stands for the least nonnegative integer $a$ satisfying $a\equiv b\pmod c$. If $\vec{x}=(x_1,\ldots,x_r)$ and $\vec{n}=(n_1,\ldots, n_r)\in\mathbb Z^r$, we define $\vec{x}\pmod{\vec{n}}=((x_1\!\!\!\mod n_1),\ldots,(x_r\!\!\!\mod n_r))$. Analogously, for any matrix $A\in\mathbb Z^{k\times k}$, we define $(A\pmod{n})_{ij}=(A_{ij}\mod n)$.

\end{definition}

The problem of recovering $x$ from the residue representation $ ((x\!\!\!\mod n_1), \ldots , $ $ (x\!\!\!\mod n_r)  ) $ is practically resolved by the following well-known theorem.
\begin{theorem}[The Chinese Remainder Theorem]
Let $n_1,\ldots,n_r$ be pairwise coprime integers, $a_1,\ldots,a_r$ be arbitrary integers, and $N=n_1\cdots n_r$. Then there exists an integer $x$ such that
\begin{equation}\label{eq04}
x\equiv a_1\pmod{n_1},\ldots, x\equiv a_r\pmod{n_r},
\end{equation}
and any two integers $x_1$ and $x_2$ satisfying (\ref{eq04}) satisfy also $x_1\equiv x_2\pmod{N}$.
\end{theorem}

\begin{remark} The above remarks and the Chinese Remainder Theorem imply that the integer ring operations $(+,\cdot)$ can be implemented using the residue representation, and that the integers can be uncovered from the residue representations provided that 1) $\vec{n}=(n_1,\ldots, n_r)$ consists of pairwise coprime integers and 2) the integers stay in interval of length $N-1$, where $N=n_1\cdots n_r$.
\end{remark}

\begin{remark} In order to ensure that $\vec{n}=(n_1,\ldots,n_r)$ consists of pairwise coprime integers, we select numbers $n_i$ from the set of prime numbers. For the reasons that will become obvious later, we will however omit the first prime $2$.
\end{remark}

\begin{definition}
    $\vec{p}_r$ is an $r$-tuple $\vec{p}_r=(3,5,7,\ldots, p_r)$ consisting of $r$ first primes by excluding $2$. For this selection, a consequence of the prime number theorem is that, asymptotically, $P_r=3\cdot 5 \cdot 7 \cdot ~  \cdots ~ \cdot p_{r}=\frac12e^{(1+o(1))r\ln r}$.
\end{definition}

\begin{theorem}[Macarie
    \cite{Macarie1998}]
    \label{thm03}
    $\SLeQ \subseteq \Log $
\end{theorem}

\begin{proof}
    For a given alphabet $ \Sigma $, let $L \in \Sigma^* $ be a language in $ \SLeQ $ and $P=(\vec{x},\{M_i\mid i\in\Sigma\},\vec{y})$ be a $k$-state rational-valued PFA over $\Sigma$ such that
\[
    L=\left\{ w\in\Sigma^*\mid f_P(w)=\frac12 \right\}.
\]
We remind that, for any input word $w=w_1\cdots w_n\in\Sigma^*$, we have
\begin{equation}
    \label{eq_pa1}
    f_P(w)=\vec{y}^T M_{w_n}\cdots M_{w_1}\vec{x}.
\end{equation}

Since each $M_i\in\mathbb Q^{k\times k}$, there exists a number $D\in \mathbb N$ providing that each
matrix $M_i'=DM_i\in\mathbb Z^{k\times k}$, and (\ref{eq_pa1}) can be rewritten as
\[
f_P(w)=\frac{1}{D^n}\underbrace{\vec{y}^T M_{w_n}'\ldots M_{w_1}'\vec{x}}_{f_{P'}(w)},
\]
and the language $L$ can be characterized as
\begin{equation}\label{eq03}
L=\{w\in\Sigma^*\mid 2f_{P'}(w)=D^n\}.
\end{equation}
Since the original matrices $M_i$ are stochastic, meaning that their entries are in $[0,1]$, it follows that each matrix $M_i'=D M_i$ has integer entries in $[0,D]$. Moreover, $f_P(w)\in[0,1]$ implies that
$f_{P'}(w)\in[0,D^n]$ for every input word $w\in\Sigma^n$.
As now $f_{P'}(w)$ can be computed by multiplying $k\times k$ integer matrices, the residue representation will serve as a space-saving technique.

We will fix $r$ later, but the description of the algorithm is as follows: For each entry $p$ of $\vec{p}_r=(3,5,7,\ldots, p_r)$, we let $M_i^{(p)}=M_i'\mod p$, and compute
\begin{equation}
    \label{eq20}
    (2f_{P'}(w)\!\!\!\mod p)=\vec{y}^TM_{w_n}^{(p)}\cdots M_{w_1}^{(p)}\vec{x}
\end{equation}
as all the products are computed modulo $p$, $k^2\log p$ bits are needed to compute (\ref{eq20}). Likewise, $(D^n\!\!\!\mod p)$ can be computed in space $O(\log p)$ for each coordinate $p$ of $\vec{p}_r$. The comparison $2f_{P'}(w)\equiv D^n \pmod p$ can hence done in $O(\log p)$ space.

Reusing the space, the comparison can be made sequentially for each coordinate of $\vec{p}_r$, and if any comparison gives a negative outcome, we can conclude that $2P'(w)\ne D^n$.

To conclude the proof, it remains to fix $r$ so that both $2f_{P'}(w)$ and $D^n$ are smaller than $P_r=3\cdot 5\cdot 7\cdot~\cdots~\cdot p_r$. If no congruence test is negative, then the Chinese Remainder Theorem ensures that $2f_{P'}(w)= D^n$.
Since $2f_{P'}(w)\le D^n$, we need to select $r$ so that
$
\frac12e^{(1+o(1))r\ln r}>2D^n,
$
which is equivalent to
$
\log\frac12+(1+o(1))r\ln r>\log 2+n\log D.
$
This inequality is clearly satisfied with $r=n$ for large enough $n$, and for each $n\ge 1$ by choosing $r=c\cdot n$, where $c$ is a positive constant (depending on $D$).

As a final remark let us note that $p_{\lfloor cn\rfloor}$, the $\lfloor cn\rfloor$-th prime, can be generated in logarithmic space and the prime number theorem implies that $O(\log n)$ bits are enough to present $p_{\lfloor cn\rfloor}$, since $c$ is a constant.
\qed\end{proof}

To extend the above theorem to cover $\SLQ$ as well, auxiliary results are used.

\begin{lemma}[Macarie \cite{Macarie1998}]
    \label{lemma01}
    If $N$ is an odd integer and $x$, $y\in [0,N-1]$ are also integers, then $x\ge y$ iff $x-y$ has the same parity as $((x-y)\!\!\! \mod N)$.
\end{lemma}

\begin{proof}
    As $x$, $y\in[0,N-1]$, it follows that
    \[
    (x-y\!\!\!\mod N)=\left\{
    \begin{array}{rl}
    x-y &\text{if $x\ge y$}\\
    N+x-y&\text{if $x<y$},
    \end{array}
    \right.
    \]
    which shows that the parity changes in the latter case since $N$ is odd.
\qed\end{proof}

The problem of using the above lemma is that, in modular computing, numbers $x$ and $y$ are usually known only by their residue representations $\res_{\vec{p}_r}(x)$ and $\res_{\vec{p}_r}(y)$, and it is not straightforward to compute the parity from the modular representation in logarithmic space. Macarie solved this problem not only for parity but also for a more general modulus (not necessarily equal to $2$).

\begin{lemma}
    [Claim modified from \cite{Macarie1998}]
    \label{lemma02}
    For any integer $x$ and modulus $\vec{p}_r=(3,5,7,\ldots, p_r)$, there is a deterministic algorithm that given $\res_{\vec{p}_r}(x)$ and $M\in\mathbb Z$ as input, produces the output $x\pmod M$ in space $O(\log p_r+\log M)$
\end{lemma}

As a corollary of the previous lemma, Macarie presented a conclusion which implies the logarithmic space simulation of rational stochastic automata.

\begin{lemma}
    [Claim modified from \cite{Macarie1998}]
    \label{lemma03}
    Let $\vec{p}_r=(3,5,7,\ldots, p_r)$ and $P_r=3\cdot 5\cdot 7 \cdot ~ \cdots ~ \cdot p_r$. Given the residue representations of
integers $x$, $y\in[0,P_r-1]$, the decisions $x>y$, $x=y$ or $x<y$ can be made in $O(\log p_r)$ space.
\end{lemma}

\begin{proof}
    The equality test can be done as in the proof Theorem \ref{thm03}, testing the congruence sequentially for each prime. Testing $x\ge y$ is possible by lemmata \ref{lemma01} and \ref{lemma02}: First compute $\res_{\vec{p}_r}(z)=\res_{\vec{p}_r}(x)-\res_{\vec{p}_r}(y)\pmod{\vec{p}_r}$, then compute the parities of $x$, $y$, $z$ using Lemma \ref{lemma02} with $M=2$.
\qed\end{proof}

The following theorem is a straightforward corollary from the above:
\begin{theorem}
    $\SLQ \subseteq \Log$.
\end{theorem}

When attempting to prove an analogous result to affine automata, there is at least one obstacle: computing the final value includes the absolute values, but the absolute value is not even a well-defined operation in the modular arithmetic. For example, $2\equiv -3\pmod 5$, but $\abs{2}\not\equiv\abs{-3}\pmod 5$. This is actually another way to point out that, in the finite fields, there is no order relation compatible with the algebraic structure.

Hence for affine automata with matrix entries of both signs, another approach must be adopted. One obvious approach is to present an integer $n$ as a pair $(\abs{n},\sgn(n))$, and apply modular arithmetic to $\abs{n}$. The signum function and the absolute value indeed behave smoothly with respect to the product, but not with the sum, which is a major problem with this approach, since to decide the sign of the sum requires a comparison of the absolute values, which seems impossible without having the whole residue representation. The latter, in its turn seems to cost too much space resources to fit the simulation in logarithmic space.

Hence the logspace simulation for automata with matrices having both positive and negative entries seems to need another approach. It turns out that we can use the procedure introduced by Turakainen already in 1969 \cite{Turakainen69A, Turakainen1969}.

\begin{theorem}
    {$\AfLQ \subseteq \Log $.}
\end{theorem}
\begin{proof}
    For a given alphabet $ \Sigma $, let $L \in \Sigma^* $ be a language in $ \AfLQ $ and $A=(\vec{x},\{M_i\mid i\in\Sigma\},F)$ be a $k$-state rational-valued AfA over $\Sigma$ such that
    \[
        L=\left\{ w\in\Sigma^*\mid f_A(w) > \frac12 \right\}.
    \]
    For each $ M_i \in \Q^{k \times k} $, we define a new matrix as
    $
        B_i=\left(
        \begin{array}{ccc}
            0 & \vec{0}^T & 0\\
            \vec{c}_i & M_i & \vec{0}\\
            e_i & \vec{d}_i^T & 0
        \end{array}
        \right),
    $
    where $\vec{c}_i$, $\vec{d}_i$, and $e_i$ are chosen so that the column and row sums of $B_i$ are zero. We define
    $
        \vec{x}'=\left(\begin{array}{c}0\\ \vec{x}\\ 0\end{array}\right)
    $
    as the new initial state. For the projection matrix $F$, we define an extension
    $
        F'=\left(\begin{array}{ccc} 0 & 0 & 0\\ 0& F & 0 \\ 0 & 0& 0\end{array}\right).
    $
    It is straightforward to see that $\abs{B_w\vec{v}_0'}=\abs{M_wv_0}$ as well as
$\abs{F'B_w\vec{v}_0'}=\abs{FM_wv_0}$.

For the next step, we introduce a $(k+2)\times (k+2)$ matrix $\E$, whose each element is $1$.
It is then clear that $\E^n=(k+2)^{n-1}\E$ and $B_i\E=\E B_i=\mathbf 0$. Now we define
\[
C_i=B_i+m\E,
\]
where $m\in\mathbb Z$ is selected large enough to ensure the nonnegativity of the matrix entries of each $C_i$. It follows that
\[
C_w=B_w+m^{\abs{w}}(k+2)^{\abs{w}-1}\E,
\]
and
\[
C_w\vec{x}'=B_w\vec{x}'+m^{\abs{w}}(k+2)^{\abs{w}-1}\E\vec{x}'.
\]
Similarly,
\[
F'C_w\vec{x}'=F'B_w\vec{x}'+m^{\abs{w}}(k+2)^{\abs{w}-1}F'\E\vec{x}'.
\]
Now
\[
\frac{\abs{FM_w\vec{v}_0}}{\abs{M_w\vec{v}_0}}
=\frac{\abs{F'B_w\vec{v}_0}}{\abs{B_w\vec{v}_0}}
=\frac{\abs{F'C_w\vec{v}_0'-m^{\abs{w}}(k+2)^{\abs{w}-1}F'{\E}\vec{x}'}}{
\abs{C_w\vec{x}'-m^{\abs{w}}(k+2)^{\abs{w}-1}{\E}\vec{x}'}}
\]
which can further be modified by expanding the denominators away: For an integer $g$ large enough
all matrices $D_i=gC_i$ will be integer matrices and the former equation becomes
\begin{equation}\label{eq01}
\frac{\abs{FM_w\vec{x}}}{\abs{M_w\vec{x}}}
=\frac{\abs{F'B_w\vec{x}}}{\abs{B_w\vec{x}}}
=\frac{\abs{F'D_w\vec{x}'-m^{\abs{w}}(k+2)^{\abs{w}-1}g^{\abs{w}+1}F'{\E}\vec{x}'}}{
\abs{D_w\vec{x}'-m^{\abs{w}}(k+2)^{\abs{w}-1}g^{\abs{w}+1}{\E}\vec{x}'}}.
\end{equation}
Hence the inequality
\[
\frac{\abs{FM_w\vec{x}}}{\abs{M_w\vec{x}}}\ge\frac12
\]
is equivalent to
\begin{eqnarray}
& &2\abs{F'D_w\vec{x}'-m^{\abs{w}}(k+2)^{\abs{w}-1}g^{\abs{w}+1}F'{\E}\vec{x}'}\notag\\
&\ge& \abs{D_w\vec{x}'-m^{\abs{w}}(k+2)^{\abs{w}-1}g^{\abs{w}+1}{\E}\vec{x}'}.\label{eq10}
\end{eqnarray}
In order to verify inequality (\ref{eq10}) in logarithmic space, it sufficient to demonstrate that the residue representations of both sides can be obtained in logarithmic space.

For that end, the residue representation of vector $\vec{a}=F'D_w\vec{x}'\in\mathbb R^{k+2}$ can be obtained in logarithmic space as in the proof of Theorem \ref{thm03}.

Trivially, the residue representation of $\vec{b}=m^{\abs{w}}(k+2)^{\abs{w}-1}g^{\abs{w}+1}F'{\E}\vec{x}'\in\mathbb R^{k+2}$ can be found in logarithmic space, as well. In order to compute the residue representation of
\[
\abs{\vec{a}-\vec{b}}=\abs{\vec{a}_1-\vec{b}_1}+\cdots+\abs{\vec{a_k}-\vec{b}_k}
\]
it is sufficient to decide whether $\vec{a}_i\ge\vec{b}_i$ holds. As the residue representations for each $\vec{a}_i$ and $\vec{b}_i$ is known, all the decisions can be made in logspace, according to Lemma \ref{lemma03}. The same conclusion can be made for the right hand side of (\ref{eq10}).
\qed\end{proof}

\section{A Non-affine Language}
As we saw in the previous section, $\AfL_{\mathbb Q}\subseteq\L$, and hence languages beyond $\L$, are good candidates for non-affine languages.\footnote{It is known that $\L\subsetneq\PSPACE$, so it is plausible that $\PSPACE$-complete languages are not in $\AfLQ$.} In this section, we will however demonstrate that the border of non-affinity may lie considerably lower: There are languages in $\L$ which are not affine.

In an earlier work \cite{HiMoYa}, we applied the method of Turakainen \cite{Turakainen1981} to show that there are languages in $\L$ which however are not contained in $\BAfL$. Here we will extend the previous result to show that those languages are not contained even in $\AfLA$. (We leave open whether a similar technique can be applied for $\AfL$.)

\begin{definition}[Lower density]
  \setlength{\belowdisplayskip}{0pt}
  \setlength{\belowdisplayshortskip}{0pt}
  \setlength{\abovedisplayskip}{0pt}
  \setlength{\abovedisplayshortskip}{0pt}

  Let $L\subseteq a^*$ be a unary language. We call \textbf{lower density} of $L$ the limit
  \[ \underline{dens}(L) = \liminf_{n\rightarrow \infty} \frac{\left|\{a^k\in L ~|~ k\leq n\} \right|}{n+1} .\]
\end{definition}

\begin{definition}[Uniformly distributed sequence]
  Let $(\textbf{x}_n)$ be a sequence of vectors in $\R^k$ and $I=[a_1, b_1) \times \dots \times [a_k, b_k)$ be an interval in $\mathbb R^k$. We define $C(I,n)$ as $C(I, n) = \left|\{ \textbf{x}_i \mod 1 \in I ~|~ 1\leq i \leq n \}\right|$.

  We say that \textbf{$(\textbf{x}_n)$ is uniformly distributed mod 1} if and only if for any $I$ of such type,
  \[ \lim_{n\rightarrow \infty} \frac{C(I,n)}{n} = (b_1-a_1)\cdots(b_k-a_k) .\]
\end{definition}

\medskip
\setlength{\belowdisplayskip}{\EquationSpaceGlobal}
\setlength{\belowdisplayshortskip}{\EquationSpaceGlobal}
\setlength{\abovedisplayskip}{\EquationSpaceGlobal}
\setlength{\abovedisplayshortskip}{\EquationSpaceGlobal}

\begin{theorem}
  \label{th:non_affine}
  If $L\subseteq a^*$ satisfies the following conditions:
  \begin{enumerate}
    \item \label{cond:densL} \underline{dens}(L) = 0.
    \item \label{cond:ud}    For all $N\in\N$, there exists $r\in\N$ and an ascending sequence $(m_i)\in\N$ such that
    $a^{r+m_iN}\subseteq L$ and for any irrational number $\alpha$, the sequence $\left( (r+m_iN)\alpha \right)$ is uniformly distributed mod 1.
  \end{enumerate}
  Then $L$ is not in $\AfLA $.
\end{theorem}

\begin{proof}
Let's assume for contradiction that $L\in \AfLA$. Then there exists an AfA $A$ with $s$ states, matrix $M$ and initial vector $\vec{v}$ such that the acceptance value of $A$ is
\begin{equation}\label{eq:acc_prob}
f_A(a^n) = \frac{\abs{PM^n\vec{v}}}{\abs{M^n\vec{v}}} .
\end{equation}
Without loss of generality, we can assume that the cutpoint equals to $\frac12$, and hence
$ w\in L \Leftrightarrow f_A(w)>\frac12 .$

Using the Jordan decomposition $M=PJP^{-1}$, one has $M^n = PJ^nP^{-1}$. So the coordinates of $M^n\vec{v}$ have the form
\begin{equation}
  \label{eq:vi_sum}
  (M^n\vec{v})_j = \sum_{k=1}^s p_{jk}(n)\lambda_k^n,
\end{equation}
where $\lambda_k$ are the eigenvalues of $M$ and $p_{jk}$ are polynomials of degree less than the degree of the corresponding eigenvalue. For short, we denote $F(n)=f_A(a^n)$, and let
$\lambda_k = \abs{\lambda_k} e^{2i\pi\theta_k}$.

When studying expression (\ref{eq:acc_prob}), we can assume without loss of generality, that all numbers $\theta_k$ are irrational. In fact, replacing matrix $M$ with $\alpha M$, where $\alpha\ne 0$ does not change (\ref{eq:acc_prob}), since
\[
\frac{\abs{P(\alpha M)^n\vec{v}}}{\abs{(\alpha M)^n\vec{v}}}
=\frac{\abs{\alpha^nPM^n\vec{v}}}{\abs{\alpha^n M^n\vec{v}}}
=\frac{\abs{PM^n\vec{v}}}{\abs{M^n\vec{v}}}.
\]
Selecting now $\alpha=e^{2\pi i\theta}$ (where $\theta\in\mathbb R$) implies that the eigenvalues of $M$ are $\lambda_k e^{2i\pi(\theta_k+\theta)}$.
The field extension $\mathbb Q(\theta_1,\ldots,\theta_s)$ is finite, and hence there is always an irrational number $\theta\notin\mathbb Q(\theta_1,\ldots,\theta_s)$. It follows directly that
all numbers $\theta_k+\theta$ are irrational. Hence we can assume that all the numbers $\theta_k$ are irrational in the first place.\footnote{Note that the new matrix obtained may not be affine, so it would be wrong to assume that all AfAs have to admit an equivalent one with only irrational eigenvalues. However, this does not affect this proof, since we do not require the new matrix to be affine, we only study the values that the fraction $\frac{\abs{P(\alpha M)^n\vec{v}}}{\abs{(\alpha M)^n\vec{v}}}
=\frac{\abs{PM^n\vec{v}}}{\abs{M^n\vec{v}}}$ take.}

By restricting to an arithmetic progression $n=r+mN$ ($m\in\mathbb N$) we can also assume that
no $\lambda_i/\lambda_j$ is a root of unity for $i\ne j$. In fact, selecting
$N=\lcm\{\ord(\lambda_i/\lambda_j)\mid \text{$i\ne j$ and $\lambda_i/\lambda_j$ is a root of unity}\}$ (\ref{eq:vi_sum}) becomes
\begin{equation}
  \label{eq:vi_sum2}
  (M^{r+mN}\vec{v})_j = \sum_{k=1}^s p_{jk}(r+mN)\lambda_k^r(\lambda_k)^{Nm}
	=\sum_{k=1}^{s'}q_{jk}(m)\mu_k^m,
\end{equation}
where $\{\mu_1,\ldots,\mu_{s'}\}$ are the distinct elements of set $\{\lambda_1^N,\ldots,\lambda_s^N\}$
Now for $i\ne j$ $\mu_i/\mu_j$ cannot be a root of unity, since $(\mu_i/\mu_j)^t=1$ would imply
$(\lambda_{i'}/\lambda_{j'})^{Nt}=1$, which in turn implies $(\lambda_{i'}/\lambda_{j'})^{N}=1$ and hence $\mu_i=\lambda_{i'}^N=\lambda_{j'}^N=\mu_j$, which contradicts the assumption $\mu_i\ne \mu_j$.

We can now write the acceptance condition $f_A(a^n)>\frac12$ equivalently as
\begin{eqnarray*}
& &f_A(a^n)>\frac12
\Leftrightarrow 2\abs{PM^n\vec{v}}>\abs{M^n\vec{v}}\\
&\Leftrightarrow &2\sum_{j\in E_a}\abs{(M^n\vec{v})_j}>\sum_{j\in E}\abs{(M^n\vec{v})_j}
\Leftrightarrow \underbrace{\sum_{j\in E_a}\abs{(M^n\vec{v})_j}-\sum_{j\in\overline{E_a}}\abs{(M^n\vec{v})_j}}_{g(n)}>0,\end{eqnarray*}

Where $E$ is the set of states of $A$, $E_a\subseteq E$ its set of accepting states, and $\overline{E_a}$ the complement of $E_a$.
According to (\ref{eq:vi_sum}), $g(n) := \sum_{j\in E_a}\abs{(M^n\vec{v})_j}-\sum_{j\in\overline{E_a}}\abs{(M^n\vec{v})_j}$ consists of combinations of absolute values of linear combination of functions of type $n^d\lambda^n$.

We say that
$n^{d_1}\lambda_1^n$ is of {\em larger order} than $n^{d_2}\lambda_2^{n}$, if
$\abs{\lambda_1}>\abs{\lambda_2}$; and in the case $\abs{\lambda_1}=\abs{\lambda_2}$, if $d_1>d_2$.
If $\abs{\lambda_1}=\abs{\lambda_2}$, we say that $n^d\lambda_1^n$ and $n^d\lambda_2^n$ and of the same order.
It is clear that if term $t_1(n)$ is of larger order than $t_2(n)$, then $\displaystyle\lim_{n\to\infty}\frac{t_2(n)}{t_1(n)}=0$.

We can organize the terms in expression (\ref{eq:vi_sum}) as
\begin{equation}\label{eq:vi_sum01}
(M^n\vec{v})_j=\sum_{k=1}^s p_{jk}(n)\lambda_k^n=\Lambda^{(N)}_j(n)+\Lambda^{(N-1)}_j(n)+\cdots+\Lambda^{(0)}_j(n),
\end{equation}
where each $\Lambda^{(m)}_j(n)$ consists of terms with equal order multiplier:
\begin{equation}\label{eq:vi_sum02}
\Lambda^{(m)}_j(n)=\sum_{k=1}^{m_j}c_{mk}n^{d_m}{\lambda_{mk}}^n
=n^{d_m}\lambda_m^n\sum_{k=1}^{m_j}c_{mk}e^{2\pi i n\theta_{mk}}
\end{equation}
(for notational simplicity, we mostly omit the dependency on $j$ in the right hand side of (\ref{eq:vi_sum02})).
Here $\lambda_m\in\mathbb R_+$ is the common absolute value of all eigenvalues $\lambda_{mk}=\lambda_me^{2\pi i \theta_{mk}}$,
and expression (\ref{eq:vi_sum01}) is organized in descending order: $\Lambda^{(N)}_j$ is the sum of terms of the highest order multiplier, $\Lambda^{(N-1)}_j$ contains the terms of the second highest order multiplier, etc. We say that $\Lambda^{(k_2)}_j$ is lower than $\Lambda^{(k_1)}_j$ if $k_2<k_1$

We will then fix a representation
\begin{eqnarray}
g(n)&=&\sum_{j\in E_a}\abs{\sum_{k=1}^sp_{jk}(n)\lambda_k^{n}}
-\sum_{j\in \overline{E_a}}\abs{\sum_{k=1}^sp_{jk}(n)\lambda_k^{n}}\notag\\
&=&\sum_{j\in E_a}\abs{A_j(n)+B_j(n)+C_j(n)}-\sum_{j\in \overline{E_a}}\abs{A_j(n)+B_j(n)+C_j(n)}\label{eq:vi_sum3},
\end{eqnarray}
where $A_j(n)+B_j(n)+C_j(n)$ is a grouping of all $\Lambda$-terms in (\ref{eq:vi_sum01}) defined as follows:

\begin{enumerate}
\item $\displaystyle A_j(n)=\sum_{k=0}^m\Lambda_j^{(N-k)}(n)$, where $m\in[-1,N]\cap\mathbb Z$ is chosen as the maximal number so that
\begin{equation}\label{eq:vi_sum4}
A=\sum_{j\in E_a}\abs{A_j(n)}-\sum_{j\in \overline{E_a}}\abs{A_j(n)}
\end{equation}
is a constant function $\mathbb N\to\mathbb R$. Such an $m$ exists, since for $m=-1$, the sum is regarded empty and $A_j(n)=0$, but for $m=N$, all $\Lambda$-terms are included, and then (\ref{eq:vi_sum4}) becomes $f_A(a^n)$, which is not constant (otherwise condition 1 or 2 of the theorem would be false).
\item $B_j(n)$ consists a single $\Lambda$-term immediately lower than those in $A_j(n)4$, and
\item $C_j(n)$ contains the rest of the $\Lambda$-terms, lower than $B_j(n)$
\end{enumerate}
\begin{lemma}
    \label{approximation}
    If $A\ne 0$, then
    $
        \forall z\in\mathbb{C},
        \abs{A+z}=\abs{A}+\re{\dfrac{\abs{A}}{A}}z+O(\dfrac{z^2}{A}) .
    $
\end{lemma}
\begin{proof}
Denote $z=x+iy$. Because $\abs{\re z}\le\abs{z}$, we have
\begin{eqnarray*}\abs{1+z}&=&\abs{1+x+iy}=\sqrt{(1+x)^2+y^2}=\sqrt{1+2\re z+\abs{z}^2}\\
&=&1+\re{z}+O(z^2).
\end{eqnarray*}
Now
\[
\abs{A+z}=\abs{A}\abs{1+\frac{z}{A}}=\abs{A}\big(1+\re\frac{z}{A}+O(\big(\frac{z}{A}\big)^2)\big)
=\abs{A}+\re{\frac{\abs{A}}{A}z}+O(\frac{z^2}{A}).
\]
\qed
\end{proof}
We choose $\lambda\in\mathbb R_+$ and $d$ so that the highest $\Lambda$-term in $B(n)$ is of order $n^d\lambda^n$ and define
$A^{\prime}_j(n)=n^{-d}\lambda^{-n}A_j(n)$, $B^{\prime}_j(n)=n^{-d}\lambda^{-n}B_j(n)$, $g^{\prime}(n)=g(n)n^{-d}\lambda^{-n}$. Then clearly $g^{\prime}(n)>0$ if and only if $g(n)>0$ and each $B_j(n)$ remains bounded as $n\to\infty$. To simplify the notations, we omit the primes and recycle the notations to have a new version of $g(n)$ of (\ref{eq:vi_sum3}) where $A_j$-terms may tend to infinity but $B_j$-terms remain bounded.

Recall that we may assume (by restricting to a arithmetic progression) that no $\lambda_i/\lambda_j$ is a root of unity. By Skolem-Mahler-Lech theorem \cite{Hansel}, this implies that functions $A_j$ can have only a finite number of zeros, and in the continuation we assume that $n$ is chosen so large that no function $A_j$ becomes zero. Furthermore, by the main theorem of \cite{Everetse}, then $\abs{A_j(n)}=\Omega(n^d\lambda^{n-\epsilon})$ for each $\epsilon>0$.\footnote{This is the only point we need the assumption that the matrix entries are algebraic.} As each $B_j$ remains bounded, we find that  $B_j^2/A_j$ tend to zero as $n\to\infty$, and hence by Lemma \ref{approximation}, defining
\begin{eqnarray*}
& &g_1(n)=\\
& &\sum_{j\in E_a}\Big(\abs{A_j(n)}+\re(\frac{\abs{A_j(n)}}{A_j(n)}B_j(n))\Big)-
\sum_{j\in \overline{E_a}}\Big(\abs{A_j(n)}+\re(\frac{\abs{A_j(n)}}{A_j(n)}B_j(n))\Big)\\
&=&\underbrace{\sum_{j\in E_a}\abs{A_j(n)}-\sum_{j\in \overline{E_a}}\abs{A_j(n)}}_{h(n)}+\sum_{j\in E_a}\re(\frac{\abs{A_j(n)}}{A_j(n)}B_j(n))
+\sum_{j\in \overline{E_a}}\re(\frac{\abs{A_j(n)}}{A_j(n)}B_j(n))
\end{eqnarray*}
we have a function $g_1(n)$ with the property $g_1(n)-g(n)\to 0$ ($C$-terms are lower than $B$-terms, so they can be dropped without violating this property), when $n\to\infty$. Also by the construction it is clear that $h(n)=C\cdot n^d\lambda^n$, where $C$ is a constant, and by the conditions of the theorem, this is possible only if $C=0$.

Notice tat $g_1(n)$ is not a constant function by construction.
Also, each $B_j$ is a linear combination of functions of form $e^{2\pi i\theta_k n}$, each $\theta_k$ can be assumed irrational, and $\abs{\abs{A_j(n)}{A_j(n)}=1}$, so we can conclude that $g_1(n)$ is a continuous function formed of terms of form $ce^{i\theta_kn}$ and of ratios $\abs{A_j}/A_j$. In these terms, however the behaviour is asymptotically determined by the highest $\Lambda$-terms, so the conclusion remains even if we drop the lower terms.

By assumption, for all $k$, the sequence $(r+mN)\theta_k$ is uniformly distributed modulo 1. It follows that the values $e^{2i\pi (r + m N)\theta_k}$ are dense in the unit circle. If for some $m$, $g_1(r+mN) < 0$, then $g_1(r+Nm) \leq -\varepsilon$ for some $\epsilon>0$.
Then, because of the density argument, there are arbitrarily large values of $i$ for which $g_1(r+m_iN)\leq 0$ contradicting condition \ref{cond:ud} of the statement. Hence $g_1(r+mN)\ge 0$ for each $m$ large enough. As $g_1$ is not a constant, there must be some
$m_0$ so that $g_1(m_0)\ge\epsilon>0$.

Next, let $R(x_1,\ldots,x_s)$ be a function obtained from $g_1$ by replacing each occurrence of $e^{i\theta_kn}$ by a variable $x_k$, hence each $x_k$ will assume its value in the unit circle. Moreover, by the assumptions of the theorem, the values of $x_k$ will be uniformly distributed in the unit circle.

Note that $g_1(n) = R((e^{2i\pi (r + m_i N)\theta_k})_{k\in A})$. Then, because the sequences $((r+m_iN)\theta_k)_i$ are uniformly distributed modulo 1, it follows that any value obtained by the function $R((e^{2i\pi y_k})_{k\in A})$ can be approximated by some $g_1(r+m_i M)$ with arbitrary precision.
The function $R$ is continuous, therefore there exists an interval $I=(x_1, y_1, ...) = ((x_k, y_k))_{k\in A}$ on which $R((x_k)) > \frac{\varepsilon}{2}$. So, if $m_i$ is large enough and satisfies
\[
  \left( (r+m_iN)\theta_1 \mod 1, \dots \right) = \left( (r+m_iM)\theta_k \mod 1 \right)_{k\in A} \in I  ,
\]
then $g_1(r+m_iN) > \frac{\varepsilon}{2}$, which implies $f_A(r+m_iN) > 0$ and hence $a^{r+m_iN}\in L$. Now we just have to prove that the sequence $(r+m_iN)$ is "dense enough" to have $\underline{dens}(L) > 0$, contradicting again condition \ref{cond:densL}.\\
Then, because of uniform distribution imposed by condition \ref{cond:ud}, one has
\[
  d = \lim_{i\rightarrow\infty} \frac{C(I, r+mN)}{r+mN} = \prod_{k\in A} (y_k - x_k)
\]
And so for $i$ large enough, $\frac{C(I, r+m_iN)}{r+m_iN} \geq \frac{d}{2}$, with $a^{h+n_iQ}\in L$, implying $\underline{dens}(L) > 0$, a contradiction.
\qed\end{proof}

\begin{corollary}
  \label{coro:poly_non_affine}
  Let $P$ be any polynomial with nonnegative coefficients and ${\deg(P) > 2}$.
  The language $\{ a^{P(n)} ~|~ n\in \N \}$ is not in $ \AfLA $.
\end{corollary}

\begin{corollary}
  \label{coro:prime_non_affine}
  The language $\{ a^p ~|~ p \text{ prime} \}$ is not in $\AfLA$.
\end{corollary}

\begin{proof}[Proof of Corollary \ref{coro:poly_non_affine} and Corollary \ref{coro:prime_non_affine}.]
  Turakainen proved that these two languages satisfies the two conditions of Theorem \ref{th:non_affine} \cite{Turakainen1981}.
  Therefore, these two languages not in $\AfLA$.
\qed\end{proof}

\section*{Acknowledgments}

Yakary\i lmaz was partially supported by Akad\={e}misk\={a} person\={a}la atjaunotne un kompeten\v{c}u pilnveide Latvijas Universit\={a}t\={e} l\={\i}g Nr. 8.2.2.0/18/A/010 LU re\c{g}istr\={a}cijas Nr. ESS2018/289 and ERC Advanced Grant MQC. Hirvensalo was partially supported by the Väisälä Foundation and Moutot by ANR project CoCoGro (ANR-16-CE40-0005).

\bibliographystyle{splncs04}
\bibliography{affine}

\begin{flushright}

\end{flushright}


\end{document}